\documentclass[aps,pra,
    twocolumn,
    longbibliography,
    superscriptaddress,floatfix]{revtex4-1} 

\pdfoutput=1
\usepackage{amsmath,amsthm,amsfonts,amssymb}
\usepackage{graphicx}
\usepackage{color}
\usepackage{braket}
\usepackage{hyperref}

\numberwithin{equation}{section}
\numberwithin{figure}{section}

\newtheorem{theorem}{Theorem}
\newtheorem{lemma}{Lemma}

\newtheorem{corollary}{Corollary}

\DeclareMathAlphabet{\mathpzc}{OT1}{pzc}{m}{it}

\newcommand{\FF}{{\mathbb{F}}}

\newcommand{\calQ}{{\mathcal{Q}}}

\begin{document}

\title{Distillation with sublogarithmic overhead}
\author{Matthew B.~Hastings}

\affiliation{Station Q, Microsoft Research, Santa Barbara, CA 93106-6105, USA}
\affiliation{Station Q Quantum Architectures and Computation, Microsoft Research, Redmond, WA 98052, USA}

\author{Jeongwan Haah}
\affiliation{Station Q Quantum Architectures and Computation, Microsoft Research, Redmond, WA 98052, USA}

\begin{abstract}
%We present an elementary construction of a family of quantum error correcting codes that allow a distillation efficiency $\gamma<1$,
%asymptotically approaching $0.6779\ldots$ The smallest example that leads to $\gamma<1$ uses a code with $2^{58}$ qubits.
It has been conjectured~\cite{bh} that for any
distillation protocol for magic states for the $T$ gate,
the number of noisy input magic states required per output magic state at output error rate $\epsilon$ 
is $\Omega(\log(1/\epsilon))$.  We show that this conjecture is false.
We find a family of quantum error correcting codes of parameters 
$[[\sum_{i=w+1}^m \binom{m}{i}, \sum_{i=0}^{w} \binom{m}{i}, \sum_{i=w+1}^{r+1} \binom{r+1}{i}]]$
for any integers $ m > 2r$, $r > w \ge 0$, by puncturing quantum Reed-Muller codes.
When $m > \nu r$, our code admits a transversal logical gate at the $\nu$-th level of Clifford hierarchy.
In a distillation protocol for magic states at the level $\nu = 3$ ($T$-gate),
the ratio of input to output magic states
is $O(\log^\gamma (1/\epsilon))$ where $\gamma = \log(n/k)/\log(d)< 0.678$ for some $m,r,w$.
The smallest code in our family for which $\gamma < 1$ is on $\approx 2^{58}$ qubits.
\end{abstract}

\maketitle
\section{Introduction}

One of the most promising paths towards a scalable quantum computer involves implementing very high accuracy Clifford operations, and using them to perform magic state distillation~\cite{Knill2004a,BravyiKitaev2005Magic}, turning a large number of noisy $T$ gates into a small number of $T$ gates with some small error $\epsilon_{out}$.  This
 magic state distillation is estimated to be the major source of overhead,
and is thus of great theoretical and practical importance.

Assuming perfect Cliffords,
three previous protocols~\cite{Jones2012,HHPW,rtrio} 
enabled magic state distillation with a ratio of input to output magic states which is $O(\log^\gamma (1/\epsilon_{out}))$
as $\epsilon_{out} \to 0$ for $\gamma$ arbitrarily close to $1$.  
It has been conjectured\cite{bh} that $\gamma \geq 1$ for all protocols.

Given an $[[n,k,d]]$ error correcting code admitting a transveral $T$ gate,
then using noisy magic states of error rate $\epsilon_{in}$ 
one can output $k$ magic states of error rate $ O((n\epsilon_{in})^d)$.
Concatenating $z$ times, one obtains $k^z$ magic states of error rate $\epsilon_{out} = O(n^{d^z} \epsilon_{in}^{d^z})$
from $n^z$ noisy magic states.
For a fixed code, the ratio of input to output magic states
is thus $O( \log ^\gamma (1/\epsilon_{out}))$
where $ \gamma = \frac{\log (n/k)}{\log d}$.  
We find quantum error correcting codes with $\gamma$ asymptotically approaching $0.6779\cdots$.

\section{Definitions, Results and Proofs}

For any non-negative integers $m \geq  r$, let $RM(r,m)$ denote the classical Reed-Muller code on $2^m$ bits;
a codeword of $RM(r,m)$ is a complete list of function values $( f(v) :  v \in \FF_2^m ) \in \FF_2^{2^m}$
where $f$ is a polynomial of degree at most $r$ in $m$ binary variables $x_i = x_i^2$.
We will not distinguish the list of function values from the function itself.
If $m > \nu r$, then every codeword of $RM(r,m)$ has weight divisible by $2^\nu$~\cite{Ward1990}.
It is also well-known that $\dim_{\FF_2} RM(r,m) = \sum_{i=0}^r \binom{m}{i} =: \binom{m}{\le r}$, 
$RM(r,m)^\perp = RM(m-r-1,m)$,
and the minimum distance of $RM(r,m)$ is $2^{m-r}$~\cite{MacWilliamsSloane}.

Let $|v|$ denote the number of 1's in $v \in \FF_2^m$ (Hamming weight).
For any integer $w < m$, let $PRM(r,m,w)$ denote the {\it punctured} Reed-Muller code
by forgetting coordinates $v$ with $|v| \le w$;
a codeword of $PRM(r,m,w)$ is a list of function values $( f(v) : v \in \FF_2^m \text{ and } |v| > w)$ 
of a degree-$r$ polynomial $f$.
The codeword length is $\sum_{i=w+1}^m\binom{m}{i} =: \binom{m}{> w}$.
If $w < 0$, then $PRM(r,m,w) = RM(r,m)$.
The dual of $PRM(m-r-1,m,w)$ is a {\it shortened} Reed-Muller code $SRM(r,m,w)$,
whose codewords are still of form $(f(v) : v \in \FF_2^m \text{ and } |v| > w )$ 
but the polynomial function has to vanish on the punctured coordinates: 
$f(v) = 0$ if $|v| \le w$~\cite{MacWilliamsSloane}.
Hence, $SRM(r,m,w) = PRM(m-r-1,m,w)^\perp \subseteq PRM(r,m,w)$.
The minimum distance of $SRM(r,m,w)$ is $2^{m-r}$ if $w < r$.

\begin{theorem}
Let $w,r,m$ be integers such that $0 \le 2w < 2r < m$.
Consider a quantum CSS code $\calQ$ whose $X$-stabilizer group is given by $SRM(r,m,w)$,
and $Z$-stabilizer group by $SRM(m-r-1,m,w)$.
Then, $\calQ$ is a $[[n=\binom{m}{> w},k= \binom{m}{\le w},d=\binom{r+1}{>w}]]$ code,
and if $m > \nu r$ for some positive integer $\nu$, 
there exists a choice of logical operators such that a transversal gate $\bigotimes_{j=1}^n \mathrm{diag}(1,\exp(2\pi i /2^\nu))$
becomes a logical operator that is the product of $\mathrm{diag}(1,\exp(-2\pi i /2^\nu))$ over all logical qubits.
\end{theorem}
\begin{proof}
The code length is obvious by definition. 
We need the following lemma before proving the values of $k,d$:
\begin{lemma}
Given $f : \FF_2^{m} \to \FF_2$, 
let $|f|_{>w}$ be the number of $v\in \FF_2^m$ such that $f(v)=1$ and $|v|>w$.
Let $D(r,m,w)=\min_{f \in RM(r,m), f \neq 0} |f|_{>w}$.

Then, $D(r,m,w)=\binom{m-r}{>w}$.
In particular, 
if $m - r > w$,
the minimum distance of $PRM(r,m,w)$ is $\binom{m-r}{>w}$,
and there is no nonzero polynomial function of degree at most $r$ supported on $\{ v \in \FF_2^{m} : |v| \le w \}$. 
\end{lemma}
\begin{proof}
The polynomial function $(1+x_1)\cdots(1+x_r)\in RM(r,m)$,
has $|f|_{>w} = \binom{m-r}{>w}$.
We show it is the minimum by induction on $m$.
The base case $m = 0$ is clear since $D(0,0,w < 0) = 1$ and $D(0,0,w \ge 0) = 0$.
Now let $m > 0$ and assume $D(r,m-1,w)=\binom{m-r-1}{>w}$ for all $0 \leq r \leq m-1$ and any $w$.

For $r=0$ or $r=m$, it is obvious that $D(r,m,w)=\binom{m-r}{>w}$.
For $0<r<m$, we use the inductive hypothesis and the recursive construction of Reed-Muller codes; namely, 
any polynomial function $f$ in $x_1,\ldots,x_m$ of $RM(r,m)$ can be written as 
$f(x_1,\ldots,x_{m}) = g(x_1,\ldots,x_{m-1}) + x_{m} h(x_1,\ldots,x_{m-1})$
where $g \in RM(r,m-1)$ and $h \in RM(r-1,m-1)$.
To find a lower bound on $|f|_{>w}$, we separate cases where $h=0$ and $h \neq 0$.
If $h=0$, then $|f|_{>w} = |g|_{>w} + |g|_{>w-1}$ where $|g|_{>w}$ is when $x_{m} = 0$ and $|g|_{>w-1}$ is when $x_{m} = 1$.
(Here, the domain of $g$ and $h$ is $\FF^{m-1}_2$.)
Hence, $|f|_{>w} \ge D(r,m-1,w) + D(r,m-1,w-1)$.
If $h \neq 0$, then by a triangle inequality we have $|g+h|_{>w} \ge |h|_{>w} - |g|_{>w}$,
implying that $|f|_{>w} = |g|_{>w} + |g+h|_{>w-1} \ge |g|_{>w} + |g+h|_{>w} \ge |h|_{>w} \ge D(r-1,m-1,w)$.
Therefore,
\begin{align}
&D(r,m,w) \nonumber \\
&\ge \min\left( \begin{matrix}D(r,m-1,w)+D(r,m-1,w-1),\\ D(r-1,m-1,w)\end{matrix} \right) \nonumber \\
&= \binom{m-r}{>w},
\end{align}
where the last equality follows
by the Pascal identity on the binomial coefficients.
\end{proof}

To find the desired set of logical operators, we represent $PRM(r,m,w)$, 
the set of all $X$-logical operators including $X$-stabilizers,
as the span of the rows of $G_T$ and $G_0$ where
\begin{align}
    \begin{bmatrix}
        I_k & G_T \\
        0 & G_0
    \end{bmatrix}
\end{align}
is the generating matrix for $RM(r,m)$ obtained 
by bringing punctured coordinates (there are $k = \binom{m}{\le w}$ of them) 
to the left by permutation of columns, and Gaussian elimination on the rows.
The fact that the top-left submatrix is the full rank identity matrix is due to the lemma,
since, otherwise, the submatrix would have a nonzero right kernel, 
which is impossible because any nonzero vector in the dual of $RM(r,m)$ is not supported on the punctured coordinates.
The desired basis of the logical operators is given by $G_T$;
declare that each row of $G_T$ corresponds to a pair of  $X$- and $Z$-logical operators.
This gives the correct commutation relations, and thus the number of logical qubits is $\binom{m}{\le w}$.

The dual of the $X$-stabilizer space is $PRM(m-r-1,m,w)$, and hence the minimum of the weight of any $Z$-logical operator 
is $\binom{r+1}{>w}$ by the lemma.
The dual of the $Z$-stabilizer space is $PRM(r,m,w)$, and hence the minimum of the weight of any $X$-logical operator 
is $\binom{m-r}{>w}$.
Thus, $d \geq \binom{r+1}{>w}$.  
In fact, $d = \binom{r+1}{>w}$ because any stabilizer belongs to either $SRM(r,m,w)$ or $SRM(m-r-1,m,w)$,
and hence has weight $\ge 2^{r+1} > \binom{r+1}{>w}$ or zero.
\iffalse
To show $d=\binom{r+1}{>w}$, we show that the product of Pauli $Z$ operators on qubits such that $(1+x_1) \cdots (1+x_r)=1$ is a nontrivial logical operator as follows: the list of function values $(1+x_1) \cdots (1+x_r)$ in $\FF_2^{2^m}$  has even inner product with every row of $\begin{bmatrix} I_k & G_T\end{bmatrix}$ and has nontrivial support on the punctured coordinates, so it has odd inner product with at least one row of $G_T$.  
\fi

The transversality of the logical operators 
can be computed easily by working with state vector directly.
See~\cite{bh}.
One should use the fact that any set of $\ell \ge 2$ distinct rows of 
$\begin{bmatrix} G_T \\ G_0 \end{bmatrix}$ have overlap that is a multiple of $2^{\nu-\ell+1}$~\cite{Ward1990}.
\end{proof}

\begin{corollary}
There exist quantum codes of parameters $[[n,k,d]]$ admitting transversal logical gate $T = \mathrm{diag}(1,e^{i \pi/4})$ 
simultaneously on every logical qubit with $\gamma = \log(n/k)/\log d$ arbitrarily close to $\gamma_0 = 0.6779..$
\end{corollary}
\begin{proof}
Take $m = 3r+1$ and $w = 3rp$ for $p \in (1/6,1/3)$.
In the large $r$ limit, $\gamma$ converges to $3(1-S(p))/S(3p)$ where $S(p) = - p \log_2 p - (1-p) \log_2 p$,
which can be seen by the Stirling approximation. At $p = 0.270629..$, we have $\gamma = 0.67799..$
\end{proof}
We have verified that the smallest code such that $m = 3r+1$ with $\gamma < 1$ has $r = 19$ and $w = 14$
so the code is $[[288215893050995568,14483100716176,21700]]$.

\section{Discussion}
We have given a code with $\gamma < 1$.
It is not clear what the infimum of $\gamma$ over all codes is; indeed, we know no proof that $\gamma$ is bounded away from zero.  The $r=19,w=14$ code is quite large, but
Ref.~\onlinecite{rtrio} used random puncturing of Reed-Muller codes followed by removing certain punctures to increase distance to find codes with $<1000$ qubits and $\gamma<1.2$, 
giving reason to hope that future work may find smaller examples with $\gamma<1$.
% It may be possible to analyze random puncturing theoretically using  \cite[Thm.~3.1]{avi},\cite[Thm.~3.1]{klp} to
%estimate the number of codewords of $RM(r,m)$ with given weight and then computing the probability that after puncturing that each such codeword becomes low weight (this probability is exponentially small in the codeword weight); we leave this for future work.

One may also ask for the infimum of $\gamma$ over codes with $k=1$.
We do not
know any such code with $\gamma < 2$ (the random triorthogonal codes of Ref.~\onlinecite{rtrio} and the protocols of Ref.~\onlinecite{HHPW} both allow $\gamma\rightarrow 2$ for $k=1$).

\bibliography{sd-ref}
\end{document}